\documentclass[11pt]{article}

\usepackage{amsmath}
\usepackage{amsthm}
\usepackage{amssymb}
\usepackage{latexsym}
\usepackage{amsfonts}
\usepackage{subfig}
\usepackage{tikz}
\usepackage{enumitem}
\usepackage{graphicx}
\usepackage[Assignment]{algorithm}
\usepackage{algorithmic}
\usetikzlibrary{calc}

\def\pr{{ \: \preceq \:}}
\def\xi{{\chi}}

\def\s{{\sigma}}
\def\D{{\Delta}}
\def\P{{P}}
\def\R{{r}}

\def\l{{\ell}}
\def\A{{\cal A}}
\def\W{{W}}

\def\rootnf{{\lceil \sqrt{n} \rceil}}
\def\rootsigma{{\lceil \sqrt{\sigma} \rceil}}

\long\def\cut#1{{}}

\newtheorem{lemma}{Lemma}
\newtheorem{theorem}{Theorem}

\begin{document}

\title{Observability of Lattice Graphs}
\author{{Fangqiu Han \qquad Subhash Suri \qquad Xifeng Yan}\\[2ex]
	Department of Computer Science \\
	University of California \\
	Santa Barbara, CA 93106, USA}
\date{}
\maketitle

\begin{abstract}

We consider a graph observability problem: \emph{how many edge colors are needed for an
unlabeled graph so that an agent, walking from node to node, can uniquely determine its
location from \emph{just} the observed color sequence of the walk?}
Specifically, let $G(n,d)$ be an edge-colored subgraph of $d$-dimensional (directed or
undirected) lattice of size $n^d = n \times n \times \cdots \times n$. We say that $G(n,d)$
is \emph{$t$-observable} if an agent can uniquely determine its current position in the
graph from the color sequence of any $t$-dimensional walk, where the dimension is the
number of different \emph{directions} spanned by the edges of the walk. A walk in an
undirected lattice $G(n,d)$ has dimension between $1$ and $d$, but a directed walk can
have dimension between $1$ and $2d$ because of two different orientations for each axis.

We derive bounds on the number of colors needed for $t$-observability.
Our main result is that $\Theta(n^{d/t})$ colors are both necessary and sufficient
for $t$-observability of $G(n,d)$, where $d$ is considered a constant.
%
%
This shows an interesting dependence of graph observability on the \emph{ratio}
between the dimension of the lattice and that of the walk. In particular,
the number of colors for \emph{full-dimensional} walks is $\Theta(n^{1/2})$ in the
directed case, and $\Theta(n)$ in the undirected case, \emph{independent} of the lattice
dimension.
All of our results extend easily to non-square lattices:
given a lattice graph of size $N = n_1 \times n_2 \times \cdots \times n_d$,
the number of colors for $t$-observability is $\Theta (\sqrt[t]{N})$.
\end{abstract}


\setcounter{page}{1}

\section{Introduction} \label{sec:intro}

Imagine an agent or a particle moving from node to node in an edge-colored graph.
During its walk, the agent only learns the colors of the edges it traverses. If after
a sufficiently long walk, the agent can uniquely determine its current node, then
the graph is called observable. Namely, an edge-colored graph is \emph{observable} if
the current node of an arbitrary but sufficiently long walk in the graph can be
uniquely determined simply from the color sequence of the edges in the walk~\cite{Jungers}.
A fundamental problem in its own right, graph observability also models the ``localization''
problem in a variety of applications including monitoring, tracking, dynamical systems
and control, where only partial or local information is available for
tracking~\cite{aslam1,lavalle,Crespi,lopez,guy,Lind,Ozveren,suri-track1}.
This is often the case in networks with hidden states, anonymized nodes, or
information networks with minimalistic sensing: for instance, observability
quantifies how little information leakage (link types) can enable precise tracking
of users in anonymized networks. As another contrived but motivating scenario,
consider the following problem of robot localization in minimally-sensed environment.

A low-cost autonomous robot must navigate in a physical environment using its own
odometry (measuring distance and angles). The sensor measurements are noisy
and inaccurate, the robot invariably accumulates errors in the estimates of its
own position and pose, and it must perform periodic relocalizations. Without global
coordinates (GPS), unique beacons, or other expensive navigation aids, this is a
difficult problem in general. In many situations, however, an approximate
relocalization is possible through inexpensive and ubiquitous sensors, such as
those triggered by passing through doors (beam sensors). Privacy or cost concerns,
however, may prevent use of \emph{uniquely identifiable} sensors on all doors or
entrances. Instead, sensors of only a few types (colors) are used as long as
we can localize the robot from its \emph{path history}. Formally, the robot's
state space (positions and poses) is a subset of ${\cal R}^d$, which is partitioned
into $N  = n \times n \times \cdots \times n$ cells, each representing a desired
level of localization accuracy in the state space. We assume access to
a minimalist binary sensor that detects the robot's state transition from one cell
to another. The adjacency graph of this partition is a $d$-dimensional lattice graph,
and our robot localization problem is equivalent to observability, where
edge colors correspond to different type of sensors on cell boundaries.

\vspace*{-0.25cm}

\paragraph{Problem Statement.}

With this general motivation, we study the observability problem for (subgraphs of)
$d$-dimensional lattices (directed and undirected), and derive upper and lower bounds
on the minimum number of colors needed for their observability. Lattices, while lacking the
full power of general graphs, do provide a tractable but non-trivial setting: their
uniform local structure and symmetry makes localization challenging but, as we will
show, their regularity allows coding schemes to reconstruct even relatively short walks.
We begin with some definitions to precisely formulate our problem and the results.

Let $G(n,d)$ denote a subgraph of $d$-dimensional regular square lattice of size
$N = n^d$. We want to color the edges of $G(n,d)$ so that a walk in the graph can
be localized based solely on the colors of the edges in the walk. The starting
node of the walk is not known, neither is any other information about the walk
except the sequence of colors of the edges visited by the walk. By localizing the
walk, we mean that its current node can be uniquely determined.

When the graph is directed, an edge has both a natural dimension and an \emph{orientation}:
dimension is the coordinate axis parallel to the edge, and orientation is the direction
along that axis (positive or negative). There are $2d$ distinct orientations in a
directed lattice, two for each axis.  When the graph is undirected, each edge has a
dimension but not an orientation.  The lack of orientation makes the observability of
undirected graphs more complicated: in fact, \emph{even if all edges have distinct
colors}, the agent can create arbitrarily long walks by traversing back and forth on
a single edge so that one cannot determine the current vertex from the color sequence
alone. Nevertheless, we show that any undirected walk that includes \emph{at least two
distinct edges} can be observed (localized).

\paragraph{Our Results.}

We show that lattice observability depends not on the \emph{length} of the walk, but
rather on the number of different directions spanned by the walk. In order to discuss
both directed and undirected graphs without unnecessary notational clutter, we use
a common term \emph{dimension} to count the number of different directions: it is the
number of distinct edge orientations in a directed walk and the number of distinct axes
spanned by an undirected walk.
In particular, we say that a \emph{directed} walk $W$ has \emph{dimension $t$} if it
includes edges with $t$ distinct orientations, for $t \leq 2d$.  An \emph{undirected} walk
$W$ has \emph{dimension $t$} if it includes edges parallel to $t$ distinct axes, for $t \leq d$.
We say that $G(n,d)$ is \emph{$t$-observable} if an agent, walking from node to node,
can uniquely determine its current position from the color sequence of any $t$-dimensional walk.

Our main result shows that $O(\sqrt[t]{N})$ colors are always sufficient for
$t$-observability of (directed or undirected) lattice graph $G(n,d)$, where $N = n^d$
is the size of the graph and $d$ is assumed to be a constant. A matching lower bound
easily follows from a simple counting argument. The upper bound proof uses a combinatorial
structure called \emph{orthogonal arrays} to construct the observable color
schemes, which may have other applications as well.
We prove the results for subgraphs of square lattices, but the bounds are easily extended
to rectangular lattices: given a lattice graph of size
$N = n_1 \times n_2 \times \cdots \times n_d$,
the number of colors for $t$-observability is $\Theta (\sqrt[t]{N})$.
An interesting implication is that for full-dimensional walks, the number of colors is
independent of $d$: it is $O(n^{1/2})$ for directed, and $O(n)$ for undirected, graphs.
%



\paragraph{Related Work.}

The graph observability problem was introduced by Jungers and Blodel~\cite{Jungers},
who show that certain variations of the problem are \emph{NP}-complete in general
\emph{directed} graphs. They also present a polynomial-time algorithm for deciding if
an edge-colored graph $G$ is observable based on the fact that the following two
conditions are necessary and sufficient:
$(i)$ $G$ does not have an \emph{asymptotically reachable} node $u$ with two outgoing edges
of the same color, where an asymptotically reachable node is one reachable by arbitrarily
long paths, and
$(ii)$ $G$ does not have two cycles with the same edge color sequence but different node orders.
These results hold only for directed graphs, and not much is known for observability
of \emph{undirected} graphs.
By contrast, our results show a universal (extremal) bound that holds for all lattices
of size $N = n^d$, and apply to both directed and undirected graphs.

The graph observability is related to a number of other concepts in dynamical
systems, including \emph{trackable graphs} where the goal is to detect and identify
dynamical processes based on a sequence of sensor observations~\cite{Crespi},
discrete event systems where the goal is to learn the state of the process based
purely on the sequence (colors) of the transitions~\cite{Ozveren}, and
the \emph{local automata} where a finite state automaton is called $(d,k)$-local
if any two paths of length $k$ and identical color sequence pass through the same
state at step $d$~\cite{Lind}. The primary objectives, however, in those papers
are quite different from the combinatorial questions addressed in our paper.


\section{Definitions and the Problem Statement}
\label{sec:defs}

A \emph{$d$-dimensional lattice graph} is one whose drawing in Euclidean $d$-space forms a
regular tiling. Specifically, such a graph of size $N = n^d$ has nodes at the integer-valued
points $(x_1, x_2, \ldots, x_d )$, for $x_i \in \{0, 1, \ldots, n-1 \}$. Two nodes can be
connected if their distance is one: in undirected graphs, there is at most one such edge,
and in directed graphs, there can be two edges with opposite orientations.
(Our results hold for any subgraph of the lattice, and do not require the full lattice.)
In this and the following three sections, we focus on \emph{directed} graphs only, and
return to undirected graphs in Section~\ref{sec:Undir}.
We use the notation $G(n,d)$ for a directed lattice graph of size $n^d$.
%
%
We say a directed walk $W$ in $G(n,d)$ is \emph{$t$-dimensional} if it includes edges with
$t$ distinct orientations, for $t \leq 2d$. An edge colored graph $G(n,d)$ is called
\emph{$t$-observable} if an agent can uniquely determine its current position in the
graph from the observed edge color sequence of \emph{any} directed walk of dimension $t$,
for $t \leq 2d$.
Figure~\ref{fig:fig1} below show a 4-observable graph (1a) and a non-observable graph 1(b).
The main focus of our paper is to derive bounds on the minimum number of colors that
suffice for $t$-observability of $G(n,d)$.
We begin with a few basic definitions and preliminaries.

\begin{figure}[htbp]
\centering
\includegraphics[angle=0, width=0.5\textwidth]{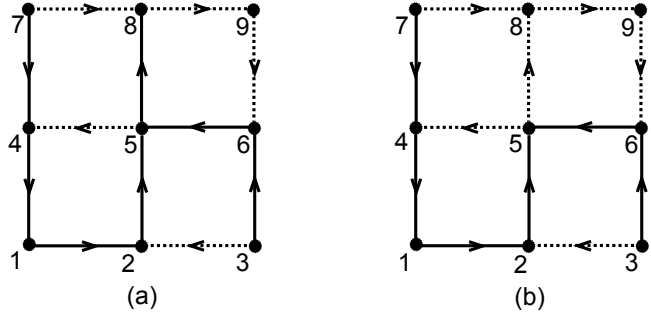}
\caption{Two nearly identical 2-colored lattice graphs, one observable (on left), the other
        non-observable (on right). The two colors are shown as \emph{solid} (S) and
        \emph{dashed} (D). Only the color of edge $(5,8)$ differs in the two graphs.
        In (b), the color sequence $(SSSD)^*$ can lead to either node 4 or 8, making
        it non-observable. In (a), any walk of dimension 4 is observable.}
\label{fig:fig1}
\end{figure}

The embedding of the lattice graph $G(n,d)$ induces a total order $\pr$ on the nodes. Let
$u = (x_1, x_2, \ldots, x_d )$ and $u' = (x'_1, x'_2, \ldots, x'_d )$ of two nodes in
the graph. Then, we say that $u \pr u'$ if $u$ precedes $u'$ in the coordinate-wise
lexicographic ordering. That is, $u \pr u'$ if either $x_1 < x'_1$, or
$x_1 = x'_1, \ldots, x_i = x'_i$ and $x_{i+1} < x'_{i+1}$, for some $1 \leq i < d$.
Given a directed walk in $G$, there is a unique \emph{minimum node} under this total
order, which we call the \emph{root} of the walk.
The node-ordering also allows us to associate edges with nodes. If $e = (u, v)$ is a
directed edge, then we say that $e$ is \emph{rooted} at $u$ if $u \pr v$, and at $v$
otherwise. (We remark that the root of an edge is unrelated to its orientation: it simply
allows us to associate edges to nodes in a unique way.) Thus, for any node $u$ in $G$,
at most $2d$ edges may be rooted at $u$: at most two directed edges in each
dimension for which $u$ is minimum under the $\pr$ order.
In Figure~\ref{fig:fig1}, for example, the edges $(1,2)$ and $(4,1)$ are
both rooted at node $1$, while both $(6,5)$ and $(5,8)$ are rooted at $5$.
The walk $(5,4,1,2,5,8)$ is rooted at node $1$.

Each edge of the graph $G$ also has a natural orientation: it is directed either
in the positive or the negative direction along its axis. To be able to refer to
this directionality, we call an edge \emph{$j$-up-edge} (resp., \emph{$j$-down-edge}) if
it has \emph{positive} (resp. negative) orientation in $j$th dimension.
In Figure~\ref{fig:fig1}, for instance, the edge $(5,8)$ is the $y$-up-edge rooted at $5$,
and $(6,5)$ is the $x$-down-edge at $5$.


\section{A Lower Bound for $t$-Observability}


\begin{theorem}	\label{thm:lb}
A directed lattice $G(n,d)$ requires at least $(n/2)^{d/t}$ colors for $t$-observability
in the worst-case, for any $t \leq 2d$.
\end{theorem}

\begin{proof}
Assume, without loss of generality, that $n$ is even. The nodes of $G(n,d)$ have
coordinates of the form $(x_1, x_2, \ldots, x_d)$, with $x_j \in \{0, 1, \ldots, n-1 \}$,
for all $j=1,2, \ldots, d$. Consider $(\frac{n}{2})^d$ unit $d$-cubes rooted at all
the \emph{even} nodes, namely, $(2x_1, 2x_2, \ldots, 2x_d)$, with
$x_j \in \{0,1, \ldots, n/2 \}$, for $j=1, 2, \ldots, d$. These unit cubes are pairwise
edge-disjoint.  We assign orientations to the some of the edges to create many $t$-dimensional
walks, and then use a counting argument to lower bound the number of colors needed for
the $t$-observability of this graph.
See Figure~\ref{fig:lb} for an illustration.
\begin{figure}[htbp]
\centering
\includegraphics[angle=0, width=0.5\textwidth]{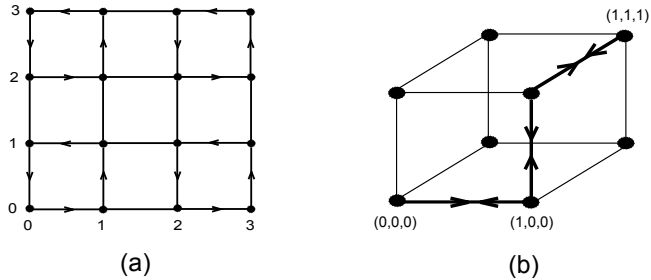}
\caption{Illustration of the lower bound. The left figure (a) shows the \emph{even} 2-cubes and
orientation of walks. The right figure (b) shows how to orient a $6$-dimensional walk for $d=3$.}
\label{fig:lb}
\end{figure}
Consider a prototypical copy of a $d$-cube, with opposite corners at
$u = (0, 0, \ldots, 0)$ and $v = (1, 1, \ldots, 1)$.
We construct a $t$-dimensional directed walk of length $t$, as follows.
Starting at $u$, for the first $\min \{t, d \}$ steps, take the $j$-up-edge in step
$j = 1, 2, \ldots, t$; for the remaining $(t-d)$ steps, take the $(2d-t+j)$-down-edges, for
$j = 1, 2, \ldots, t-d$. This construction assigns directions to $t \leq 2d$
edges of the $d$-cube; the remaining edges can be directed arbitrarily.
(Figure~\ref{fig:lb}b illustrates the construction for $d=3$ and $t=6$.)
%
By repeating the construction at all $d$-cubes rooted at the even nodes of $G(n,d)$,
we get $(n/2)^d$ disjoint $t$-dimensional walks, which must have pairwise distinct color
sequences for $t$-observability.
Since $k$ colors can produce at most $k^{t}$ distinct color sequences of length $t$,
the minimum number of colors $k$ satisfies $k^{t} \:\geq\: (\frac{n}{2})^d$, from which it follows
that $k \:\geq\: (n/2)^{d/t}$.
It is easy to see that this argument holds for undirected walk as well, where we just
have to ensure that $t \leq d$. This completes the proof.
\end{proof}

We now describe the main result of the paper, an upper bound for $t$-observability.
In order to build some intuition for the proof, and explain the coloring scheme,
we first consider a much simpler special case: the 2-dimensional lattice $G(n,2)$
and full-dimensional walks, namely, $t=4$, where we show that $O(\sqrt{n})$ colors suffice.
While the coloring and the decoding
techniques for the general case are somewhat different, this special case is
useful to explain the main ideas.

\section{Observability in $2$-Dimensional Lattices}

In discussing the two-dimensional lattice, we name the two coordinate axes $x$
and $y$, instead of $x_1, x_2$. Similarly, we use the more natural and visual
use of left-right and up-down for the directionality of edges; i.e., we say
left (resp. right) edge instead of $1$-up (resp. $1$-down) edge, and
up (resp. down) edge instead of $2$-up (resp. $2$-down) edge.  We begin with a
discussion of our coloring scheme, and then show its correctness.  We will use
4 blocks of colors, one each for \emph{left, right, up} and \emph{down} types of
edges, where each block has $\rootnf$ colors.
The coloring depends on the position of the root node $u$ associated with the edge,
and uses the quotient and the remainder of $u$'s coordinates \emph{modulo} $\rootnf$.
We use the notation $x \div m$ to denote the quotient, and $x \bmod m$ to
denote the remainder.

Specifically, consider a node $u = (x_u, y_u)$ in $G(n,2)$. There are at most 4 edges
rooted at $u$: a right outgoing edge, a left incoming edge, an up outgoing edge, and a
down incoming edge. The following algorithm assigns colors to all edges \emph{rooted} at
node $u = (x_u, y_u)$.

\paragraph{{\sc Color2}$(u)$:}

\begin{itemize} 	\advance\itemsep by -4pt
\item If $e$ is the outgoing \emph{right} edge of $u$, give it color $(x_u \div  \rootnf)$;\\
	if it is outgoing \emph{up} edge, give it color $(y_u \div \rootnf) + 2\rootnf$.
\\

\item If $e$ is the incoming \emph{left} edge of $u$, give it color $(x_u \bmod \rootnf) + \rootnf$;\\
	if it is incoming \emph{down} edge, give it color $(y_u \bmod \rootnf) + 3\rootnf$.



\end{itemize}


\cut{
Intuitively, all rightward edges receive one of the first $\rootnf$ colors based on the
\emph{quotient} of the root node's $x$ coordinate. (Recall that $x_u < n$.)
The leftward edges receive a color
indexed by $x_u \bmod \rootnf$, and the additive term $\rootnf$ ensures that these
colors come from the second batch of $\rootnf$ colors. Similarly, the up and down
edges receive colors based on the $y$ coordinate of $u$, and these colors are drawn
from 3rd and 4th batch of $\rootnf$ colors.
}
The following lemma is easy.

\begin{lemma}
All edges of $G(n,2)$ are colored, using at most $4 \rootnf$ colors.
No two outgoing edges of a node are assigned the same color.
\end{lemma}

A simple but important fact throughout our analysis is the following
\emph{tracing lemma}.

\begin{lemma} {[Tracing Lemma]} \label{lem:tracing2}
Let $W$ be a walk in $G(n,2)$ under the coloring scheme {\sc Color2}. Fixing the
position of any node of $W$ leads to a \emph{unique} embedding of $W$ in $G(n,2)$.
\end{lemma}
\begin{proof}
In our coloring scheme, the colors are grouped by the \emph{direction} of edges
(left, right, up and down), and so given the color of an edge, we know its direction.
Once a node of the walk is fixed, all subsequent (and preceding) edges are uniquely
mapped in the lattice graph.
\end{proof}

By the Tracing Lemma, our color sequence uniquely specifies the ``trace''
(or, shape) of the walk's embedding, and once we localize any node of the walk, we can
determine the embedding of the entire walk. Thus, the main problem, which will consume
the rest of the paper, is to decode the position of one node of the walk from the
color sequence. For $G(n,2)$ and $t=4$, we will
use the following simple lemma, which the more complex coloring scheme of
Section~\ref{sec:upper-dir} \emph{does not} need.

\begin{lemma} {[Pairing Lemma]} \label{lem:pairing}
Suppose a directed walk $W$ in $G(n,d)$ contains both up- and down-edges for some dimension $j$.
Then $W$ must contain a $j$-up and a $j$-down edge that are both rooted at vertices with
the same $j$th coordinate.
%
\end{lemma}
\begin{proof}
Mark each edge of $W$ parallel to the $j$th axis either $+$ or $-$ depending on
whether it is an up- or a down-edge. Since $W$ contains both a $j$-up and a $j$-down edge,
the sign changes at least once. Assume without loss of generality that it changes
from $+$ to $-$, with $e$ and $f$ being the edges associated with them.  None of the
(unmarked) edges between $e$ and $f$ are parallel to the $j$th axis, so both $e$ and
$f$ project to the same unit interval $(x_j, x_j + 1)$ on the $j$th axis. By convention,
both $e$ and $f$ are rooted at vertices whose $j$th coordinate is $x_j$.
This completes the proof.
\end{proof}

We can now explain our decoding scheme.

\begin{lemma}	\label{lem:upper2D}
Suppose $G(n,2)$ is colored using {\sc Color2}, and $\W$ is a $4$-dimensional walk
in this graph. Then, the color sequence of $\W$ uniquely determines the position of
the root node of $\W$ in the lattice.
\end{lemma}
\begin{proof}
Let $u$ be the root node of $W$, and let $(x_u, y_u)$ be its coordinates.  (These coordinates
are precisely what we want to infer from the color sequence.) By the Pairing Lemma, since
$\W$ includes all four orientations, it has two oppositely oriented edges $e_1$ (positive)
and $e'_1$ (negative), both parallel to the $x$-axis and rooted at vertices with the same $x$
coordinate. Let $v = (x_1, y_1)$ and $v' = (x_1, y'_1)$, respectively, be the root vertices
of $e_1$ and $e'_1$. By the color assignment, $e_1$ and $e'_1$, respectively, receive
colors $(x_1 \div \rootnf)$ and $(x_1 \bmod \rootnf) + \rootnf$, which together are
sufficient to uniquely calculate $x_1$. Since the edge colors uniquely determine the
edge directions (dimension and orientation), we can trace $\W$ from $e_1$ to find the
correct value of $x_u$, the $x$ coordinate of the
root node $u$. Similarly, since $W$ also includes edges with both $y$-orientations,
we can calculate $u_y$, thus uniquely localizing the root node $u$. This completes the proof.
%
\end{proof}

\begin{theorem}	\label{thm:upper2D}
$O(n^{1/2})$ colors suffice for $4$-observability of a directed lattice $G(n,2)$.
\end{theorem}
\begin{proof}
By Lemma~\ref{lem:upper2D}, the color sequence of $W$ uniquely determines the position
of $W$'s minimum node (root) in the lattice. Once $u$ is localized, the tracing lemma can
construct the unique embedding of $W$ in $G(n,2)$, localizing all other nodes, including
the current node.
\end{proof}

\section{$t$-Observability of Directed Lattices}
\label{sec:upper-dir}

We now describe the general coloring scheme for $t$-observability of $G(n,d)$, for any
fixed $d$ and all $t \leq d$. The scheme uses a tool from combinatorial design, called
\emph{orthogonal array}. We first describe our orthogonal array construction, and then
explain its application to observability.

\subsection{Orthogonal Arrays}

Let $\s, t, d$ be positive integers, with $t \leq d$. The parameters $t$ and $d$ are
in fact the walk and the lattice dimensions, while $\s$ is chosen as the smallest
\emph{prime} larger than $n^{d/t}$. (There always exists a prime between $m$ and $2m$,
for any integer $m$, and therefore $\s < 2n^{d/t}$.)
Our choice of $\s$ ensures that $\s^t \geq n^d$, which will be useful later because
each row of the array is used to assign colors to the edges rooted at a distinct node
of $G(n,d)$. An \emph{$(\s, t, d)$-orthogonal array} $\A$ is
an array of size $\s^t \times d$ satisfying the following two properties:

\begin{itemize}	\advance\itemsep by -4pt
\item The entries of $\A$ are integers from the set $\{0, 1, \ldots, \s -1 \}$, and
\\
\item For any choice of $t$ columns in $\A$, the rows (ordered $t$-tuples) are unique.
\end{itemize}

Several methods for constructing orthogonal arrays are known~\cite{OA,bush}.
Our construction uses polynomials of degree less than $t$ over the Galois field
$GF(\s) = \{0, 1, \ldots, \s -1 \}$. In particular, consider the set of all
polynomials of degree less than $t$.  Each such polynomial can be written as
$\P(x) : \sum_{i=1}^t a_i x^{t-i}$, with coefficients $a_i$ ranging over $GF(\s)$.
(Throughout the paper, when we use the word polynomial, we always mean these polynomials
over $GF(\s)$.)
There are exactly $\s^t$ such polynomials, and we can order them using the
lexicographic order of their coefficients. More specifically, let
$(a_1, a_2, \ldots, a_t )$ and $(b_1, b_2, \ldots, b_t )$ be two $t$-tuples from
$GF(\s)$, and let $\P$ and $\P'$ denote the polynomials associated with these coefficients.
Then, $\P \;\pr\; \P'$ iff $(a_1, a_2, \ldots, a_t) \;\pr\; (b_1, b_2, \ldots, b_t)$
under the lexicographic order. In particular, the polynomial for $(0, 0, \ldots, 0)$ is
the first element in this order, and the polynomial for $(\s -1, \s -1, \ldots, \s -1 )$
the last.

We will frequently need this ordering of the polynomials, and so for ease of reference,
let us define \emph{$p$-index}, which gives the unique position of a polynomial in
the ordered list. Specifically, if a polynomial has coefficients $(a_1,a_2,\ldots,a_t)$,
then its $p$-index equals $\sum_{i=1}^t a_i \s^{t-i}$. We can now describe how to
construct our orthogonal array. The $(i,j)$ entry of the array is defined as follows:

\begin{equation}
\label{eq:OA}
\A_{ij} \:\:=\:\:  \P_i (j) \bmod \s ,
\end{equation}
where $\P_i (x)$ is the polynomial whose $p$-index is $i$, for $0 \leq i < \s^t$,
and the array entry is the evaluation of this polynomial at $x = j$ (mod $\s$), where
$1 \leq j \leq d$.
More precisely, if the $i$th polynomial has coefficients $(a_1, a_2, \ldots, a_t)$, then
$$\A_{ij}  \:=\:   a_1 j^{t-1} + a_2 j^{t-2} + \cdots + a_{t-1} j + a_t  \pmod{\s} $$

In Figure~\ref{fig:OA}, we show an example of the orthogonal array constructed by our
scheme for $(\s, 2, d)$.

\begin{figure}[htbp]
\begin{center}
\begin{tabular}{r|r|r|r|r|}
\multicolumn{1}{c}{index}&\multicolumn{1}{c}{$1$}&\multicolumn{1}{c}{$2$}&\multicolumn{1}{c}{$\dots$}&\multicolumn{1}{c}{$d$}\\ \cline{2-5}
$0$ &$0$ &$0$ &$\dots$&$0$\\ \cline{2-5}
$1$ &$1$ &$1$ &$\dots$&$1$\\ \cline{2-5}
$\dots$&$\dots$&$\dots$&$\dots$&$\dots$\\ \cline{2-5}
${\s-1}$&$\s-1$ &$\s-1$ &$\dots$&$\s-1$\\ \cline{2-5}
${\s}$ &$1$ &$2$ &$\dots$&$d$\\ \cline{2-5}
${\s+1}$ &$2$ &$3$ &$\dots$&$d+1$\\ \cline{2-5}
$\dots$&$\dots$&$\dots$&$\dots$&$\dots$\\ \cline{2-5}
${2\s-1}$ &$0$ &$1$ &$\dots$&$d-1$\\ \cline{2-5}
$\dots$&\multicolumn{4}{|c|}{$\dots$}\\ \cline{2-5}
${\s^2-\s}$ &$\s-1$ &$\s-2$ &$\dots$&$\s-d$\\ \cline{2-5}
${\s^2-\s+1}$ &$0$ &$\s-1$ &$\dots$&$\s-d+1$\\ \cline{2-5}
$\dots$&$\dots$&$\dots$&$\dots$&$\dots$\\ \cline{2-5}
${\s^2-1}$ &$\s-2$ &$\s-3$ &$\dots$&$\s-d-1$\\ \cline{2-5}
\end{tabular}
\end{center}
\caption[]{A $(\s, 2, d)$ orthogonal array. For any two columns, all rows (ordered tuples)
are distinct. The $i$th row entries are computed from the polynomial $a_1 x + a_2$,
where $i$ is the $p$-index associated with $(a_1,a_2)$. Row $0$ has all zeroes because it
belongs to the polynomial $0x + 0$, which evaluates to $0$ for $j=1,2, \ldots, d$.
The last row belongs to the polynomial $(\s -1 )x + (\s -1 )$, with
coefficient vector $(\s -1, \s -1)$.  Its first entry, for $j=1$, is $(\s -1) 1 + (\s -1)
\equiv 2(\s -1) \equiv \s -2 \pmod{\s}$.}
\label{fig:OA}
\end{figure}

The following lemma shows that the construction is valid.

\begin{lemma}	\label{lem:OA}
The array $\A$ constructed by Equation~(\ref{eq:OA}) is an orthogonal array.
\end{lemma}

\begin{proof}
The array has dimensions $\s^t \times d$, and its entries come from the set
$\{0, 1, \ldots, \s - 1\}$, by construction. Thus, we only need to show that within
any $t$ columns of $\A$, all rows are distinct. We prove this by contradiction.
Let $j_1,j_2,\dots,j_t$, for $1 \leq j_1 < j_2 < \dots < j_t \leq d$ be any $t$ columns
of $\A$, and suppose that two different rows with indices $i_1$ and $i_2$, for
$i_1 < i_2$ are identical over these columns.
Let $(a_1, a_2, \ldots, a_t)$ and $(b_1, b_2, \ldots, b_t)$ denote the coefficients
corresponding to polynomials used for rows $i_1$ and $i_2$, respectively. Then,
by Equation~(\ref{eq:OA}), the polynomial used to construct entries of row $i_1$ is
$\P_{i_1} : a_1 x^{t-1} + a_2 x^{t-2} + \cdots + a_t$, and the polynomial used to
construct entries of row $i_2$ is $\P_{i_2} : b_1 x^{t-1} + b_2 x^{t-2} + \cdots + b_t$.
If these rows are identical, then we must have
$\P_{i_1} (j_k) \equiv \P_{i_2} (j_k) \pmod{\s}$, for $k=1, 2, \ldots, t$.  This implies that
$j_1, j_2, \ldots, j_t$ are $t$ distinct roots of the equation
$\P_{i_1} (x) - \P_{i_2} (x) \pmod{\s}$, which is not possible since this polynomial
has degree $t-1$ and at most $t-1$ distinct roots.\footnote{%
        Finite fields belong to unique factorization domains, and therefore a polynomial of
        order $r$ over finite fields has a unique factorization, and at most $r$ roots.}
Therefore, the rows $i_1$ and $i_2$ are not identical over the chosen $t$ columns, proving
that $\A$ is an orthogonal array.
\end{proof}

Note that not all orthogonal arrays could help us on coloring. Here we carefully constructed $\A$ such
that the regular structure of $\A$ helps us map colors to edges of the lattice graph in such
a way that a small number of appropriate colors can be used to determine the position of
a node in the lattice.

\subsection{Color Assignment using the Orthogonal Array}

We begin by indexing the nodes of $G(n,d)$ in the lexicographic rank order of their
coordinates. Specifically, a node $u$ with coordinates $(x_1,x_2,\ldots,x_d)$ has
\emph{node rank} $\R(u) = \sum_{i=1}^d  x_i n^{d-i}$, where recall that each
$x_j \in \{0, 1, \ldots n-1 \}$.
This ordering assigns rank $0$ to the origin $(0,0, \ldots, 0)$, and rank
$n^d -1 \:=\: (n-1) \sum_{i=1}^d n^{d-i}$ to the anti-origin $(n-1, n-1, \ldots, n-1)$.
Our orthogonal array $\A$ satisfies $\s^t \geq n^d$, and thus we can uniquely associate
the node with rank $i$ to the $i$th row of $\A$.

Let $u$ be a node of $G(n,d)$ whose rank (lexicographic order) is $\R(u) = i$, where
$0 \leq i < n^d$. Then, we use the $i$th row of $\A$ to assign to colors to the edges
rooted at $u$.
The rules for assigning colors are described in the following algorithm.  The algorithm uses
$2^t$ groups of disjoint colors $C_0, C_1, \ldots, C_{2^t -1}$, each with $2d \times \sigma$
colors. Each edge's color depends on its orientation, so we assign integers $1, 2, \ldots, 2d$
to the $2d$ orientations. (Any such labeling will suffice but, for the sake of concreteness,
we may number the $j$-up orientation as $j$ and the $j$-down orientation as $j+d$,
for $1 \leq j \leq d$.)

\cut{
\paragraph{{\sc ColorD}$(u)$:}

\begin{itemize}	\advance\itemsep by -4pt

\item Let $(a_1, a_2, \ldots, a_t )$ be the unique coefficient vector of a polynomial whose
	$p$-index $i$ equals the rank of $u$, namely, $i = \R(u)$.

\item Let $m \;=\; \sum_{k=1}^t (a_k \bmod 2) 2^{t-k}$.

\item If $e$ is the outgoing \emph{up} edge of $u$ in the $j$th dimension, then give
        it color \\
        $(A_{ij} \div  \rootsigma) \:+\: 2(j-1) \rootsigma$ from the color group $C_m$.

\item If $e$ is the incoming \emph{down} edge of $u$ in the $j$th dimension, then give
        it color \\
        $(A_{ij} \bmod \rootsigma) \:+\: 2(j-1) \rootsigma \:+\: \rootsigma$ from the color group $C_m$.
\end{itemize}
}

\bigskip

\paragraph{{\sc ColorD}$(u)$:}

\begin{itemize} 

\item Let $(a_1, a_2, \ldots, a_t )$ be the unique coefficient vector of a polynomial whose
        $p$-index $i$ equals the rank of $u$, namely, $i = \R(u)$.

\item Let $m \;=\; \sum_{k=1}^t (a_k \bmod 2) 2^{t-k}$.

\item If $e$ has orientation $j$, then give it color $A_{ij} \:+\: (j-1) \sigma$ from
	the color group $C_m$.

\end{itemize}

\bigskip

The use of disjoint color groups $C_m$ is required to allow unique decoding of the
position of a node in the lattice from the edge colors of
the walk. (These groups are used critically in the proof of Lemma~\ref{lem:rank2A}.)
The following two lemmas follow easily from the color assignment.

\begin{lemma}
{\sc ColorD} assigns colors to all the edges of $G(n,d)$, uses $O(n^{d/t})$
colors, and no two outgoing edges of a node receive the same color.
\end{lemma}


\begin{lemma} {[$d$-Dimensional Tracing Lemma]} \label{lem:tracingd}
Let $G(n,d)$ be colored using the scheme {\sc ColorD}, and let $W$ be a walk in $G(n,d)$.
Fixing the position of any node of $W$ leads to a \emph{unique} embedding of $W$ in $G(n,d)$.
\end{lemma}


\begin{lemma}	\label{lem:difference}
Let $p_1, p_2$, respectively, denote the $p$-indices of polynomials with coefficient
vectors $(a_1, \ldots, a_t)$ and $(b_1, \ldots, b_t)$.
Then, the quantity $A_{p_1 j} - A_{p_2 j}$ can be uniquely calculated, for any $j$,
from the $t$ coordinate \emph{differences}, namely, $a_k-b_k$, for $k=1,2,\dots,t$.
\end{lemma}
\begin{proof}
By construction, $A_{p_1 j}-A_{p_2 j} \:\equiv\: \P_{p_1}(j) - \P_{p_2}(j)  \pmod{\s}
\:\equiv\: \sum_{k=1}^t (a_k-b_k)j^{t-k} \pmod{\s}$.
\end{proof}

If $p_1$ and $p_2$, where $p_1 > p_2$, are the $p$-indices of $(a_1, a_2, \ldots, a_t)$  and
$(b_1, b_2, \ldots, b_t)$, then by definition
$p_1 = \sum_{k=1}^t a_k \s^{t-k}$ and $p_2=\sum_{k=1}^t b_k \s^{t-k}$.
Let $\ell = (p_1 - p_2)$ be the distance between these $p$-indices, and
let $(c_1, c_2, \ldots, c_t)$, with $c_i \in GF(\s)$, be the coefficient vector
that yields the $p$-index $\ell$.
Since $\ell \:=\: \sum_{k=1}^t c_k \s^{t-k}$, we can easily find each
coefficient $c_k$ from $\ell$ by modular arithmetic:
\[ c_k \:=\: (\ell \bmod \s^{t-k+1}) \div \s^{t-k} . \]

The following two lemmas establish important properties of these coefficients,
and are key to inferring \emph{locations from distances}.

\begin{lemma}	\label{lem:parity}
Let $(a_1, a_2, \ldots, a_t )$ and $(b_1, b_2, \ldots, b_t )$ be coefficient vectors
with $p$-indices $p_1, p_2$, for $p_1 > p_2$, respectively. Let $\ell = p_1-p_2$, and suppose
$(c_1, c_2, \ldots, c_t )$ is the coefficient vector of the polynomial with $p$-index $\ell$.
Then, for each $k = 1,2, \ldots, t$, we either have $a_k-b_k \equiv c_k \pmod{\s}$, or we
have $a_k-b_k \equiv c_k+1 \pmod{\s}$.
\end{lemma}
\begin{proof}
Because $\ell = p_1-p_2$, we have $\ell = \sum_{k=1}^t (a_k - b_k) \s^{t-k}$.
Perform a $\pmod{\s^{t-k+1}}$ operation on both sides of the equality:
$(a_1-b_1)\s^{t-1}+\dots+(a_t-b_t) \;\equiv\; c_1\s^{t-1}+\dots+c_t$.
We get $(a_k-b_k)\s^{t-k}+\dots+(a_t-b_t) \;\equiv \; c_k\s^{t-k}+\dots+c_t
\pmod{\s^{t-k+1}}$, which is
equivalent to $(a_k-b_k-c_k)\s^{t-k} \;\equiv\; \sum_{i=k+1}^t (c_i-a_i+b_i) \s^{t-i}
\pmod{\s^{t-k+1}}$.
Because $a_i, b_i, c_i \in GF(\s)$, the right hand side clearly satisfies
the following bounds:
\[ -\s^{t-k} \:<\: \sum_{i=k+1}^t (c_i-a_i+b_i) \s^{t-i} \:<\: 2\s^{t-k} . \]
But since $a_k,b_k,c_k$ on the left side are positive integers, there are only two feasible
solutions: $\sum_{i=k+1}^t (c_i-a_i+b_i)\s^{t-i} \equiv 0  \pmod{\s^{t-k+1}}$
or $\sum_{i=k+1}^t (c_i-a_i+b_i)\s^{t-i} \equiv \s^{t-k} \pmod{\s^{t-k+1}}$.
Therefore, we must have either
$a_k-b_k-c_k \equiv 0 \pmod{\s}$, or $a_k-b_k-c_k \equiv 1 \pmod{\s}$.
This completes the proof.
\end{proof}

The following lemma shows that while we cannot reconstruct the unknown coefficient
vectors $(a_1, a_2, \ldots, a_t)$ and $(b_1, b_2, \ldots, b_t)$ from the edge colors,
we can still compute enough information about their entries in the orthogonal
array $\A$.

\begin{lemma}	\label{lem:rank2A}
Let $G(n,d)$ be colored using the scheme {\sc ColorD}, and let $W$ be a walk in $G(n,d)$.
Let $e_1$ and $e_2$ be two edges in this walk rooted, respectively, at $u_1$ and $u_2$
with ranks (lexicographic order) $\R(u_1) = r_1$ and $\R(u_2) = r_2$, with $r_1 < r_2$.
Then, the difference $r_1-r_2$ along with the colors of $e_1$ and $e_2$ are sufficient
to compute $A_{r_1 j} - A_{r_2 j} \pmod{\s}$.
\end{lemma}
\begin{proof}
Suppose that the colors of $e_1$ and $e_2$ belongs to color groups $C_{m_1}$ and $C_{m_2}$,
respectively. Let us consider the binary representations of $m_1, m_2$, namely,
$m_1 = \sum_{k=1}^t (a_k \bmod 2) 2^{t-k-1}$ and $m_2 = \sum_{k=1}^t (b_k \bmod 2) 2^{t-k-1}$.
If $(a_k - b_k) \equiv 0 \pmod{2}$, then by Lemma~\ref{lem:parity} we can conclude
that $(a_k - b_k) \equiv c_k \pmod{\s}$ if $c_k$ is even; otherwise it equals $c_k +1$.
Similarly, $(a_k - b_k) \equiv 1 \pmod{2}$ tells us that
$(a_k - b_k) \equiv c_k \pmod{\s}$ if $c_k$ is odd (and $c_k +1$ otherwise).
Because we can infer from the embedding whether $a_k > b_k$, we can calculate
$a_k - b_k$ from $a_k - b_k \pmod{\s}$.
Once we know all $a_k - b_k$, for $k=1,2, \ldots, t$ we can
compute $A_{r_1 j} - A_{r_2 j} \pmod{\s}$
using Lemma~\ref{lem:difference}.
\end{proof}

\begin{lemma} {[Ranking Lemma}] 	\label{lem:rank}
Let a lattice graph $G(n,d)$ be colored using the algorithm {\sc ColorD}, and let
$u_1, u_2$ be two nodes in a walk $W$ of $G(n,d)$. We can calculate the difference of
their ranks $\R(u_1) - \R(u_2)$ from just the color sequence of $W$.
\end{lemma}
\begin{proof}
Using the colors of edges in $W$ starting at $u_1$, we can trace the walk, and count
the number of edges in each dimension. Suppose the number of $j$-up and $j$-down edges
included in the walk from $u_1$ to $u_2$ are $\alpha_j$ and $\beta_j$, respectively.
Then, the $j$th coordinate of $u_2$ differs from that of $u_1$ by precisely
$\alpha_j - \beta_j$. We can, therefore, calculate
$\R(u_1) - \R(u_2) \:=\: \sum_{j=1}^d (\beta_j-\alpha_j) n^{d-j}$.
\end{proof}

We are now ready to prove our main theorem.

%

\begin{theorem} \label{thm:upperD}
$O(n^{d/t})$ colors suffice for $t$-observability of any directed lattice graph $G(n,d)$,
for any fixed dimension $d$ and $t \leq 2d$.
\end{theorem}
\begin{proof}
We color the graph $G(n,d)$ using {\sc ColorD}, and show how to localize a walk $W$
of dimension $t$. Suppose $j_1,j_2,\dots,j_t$, where $1\leq j_1 < j_2<\dots<j_t\leq 2d$,
are the $t$ distinct edge orientations of $W$. Let $e_1$ be a $j_1$-oriented edge in $W$,
rooted at a node $u_1$, and let $\R(u_1)=r_1$ be the rank of $u_1$. Then, the color of
$e_1$ can be used to uniquely determine the value $\A_{r_1 j_1}$.

Now, suppose the root of the walk $\W$ is the node $u^*$, with $r^* = \R(u^*)$.
We can compute the difference $r_1 - r^*$, by using the Ranking Lemma.
By Lemma~\ref{lem:rank2A}, this difference along with the color sequence of $W$
gives us $\A_{r_1 j_1}-\A_{r^* j_1}$, from which we can calculate $\A_{r^* j_1}$.
By repeating this argument for each of the $t$ dimensions spanned by $W$, we can
calculate $\A_{r^* j_1}, \A_{r^* j_2}, \ldots, \A_{r^*, j_t}$. By the property
of orthogonal arrays, the ordered $t$-tuple $(\A_{r^* j_1}, \A_{r^* j_2},\dots, \A_{r^* j_t})$
is unique in $\A$. Therefore, the rank of the root node $u^*$ can be uniquely determined,
which in turn uniquely localizes $u^*$ in the lattice graph $G(n,d)$. This completes the proof.
\end{proof}

\cut{
\begin{proof}
We color the graph $G(n,d)$ using {\sc ColorD}, and show how to localize a
walk $W$ of cycle dimension $t$. Without loss of generality, let
$j_1,j_2,\dots,j_t$, where $1\leq j_1 < j_2<\dots<j_t\leq d$, be the dimensions
spanned by the cycles of $W$. Let $e_1$ be a $j_1$-dimensional \emph{up-edge} in one
of the cycles of $W$. Then, by the Pairing Lemma, there also exists a $j_1$-down edge,
say, $e_2$ in $W$. (We can find these paired edges for each dimension by a simple
linear-time scan of the walk.)
Suppose that the edges $e_1,e_2$ are
rooted at nodes $u_1$ and $u_2$, respectively, whose ranks are $r_1 = \R (u_1)$ and
$r_2 = \R(u_2)$, and $r_1 > r_2$.
(The proof is similar when $e_1$ is a down edge or $r_2 > r_1$.)
Then, by the Ranking Lemma, we can calculate
the difference $(r_1-r_2)$, by tracing the color sequence of $W$ from $u_1$ to $u_2$.
By Lemma~\ref{lem:rank2A}, this rank difference $r_1-r_2$ together with the color
assignment of $W$ lets us calculate the quantity $\A_{r_1 j_1} - \A_{r_2 j_1}$.
By combining these rank differences from $t$ different direction pairs, we will show
how to uniquely determine the rank of the root node $u$ of $W$.

Let $\D = \A_{r_1 j_1} - \A_{r_2 j_1}$.
Write $\A_{r_1 j_1}=\alpha_1\rootsigma+\beta_1$, $\A_{r_2 j_1}=\alpha_2\rootsigma+\beta_2$
and $\D = \alpha\rootsigma+\beta$, where $0\leq \alpha_i, \beta_i, \alpha, \beta < \rootsigma$.
First, we can calculate $\alpha$ and $\beta$ from $\D$ as follows:
$\alpha = \D \div \rootsigma$ and $\beta = \D \mod \rootsigma$.
We can write the color assigned to the edge $e_1$ by {\sc ColorD}
as $(\A_{r_1j_1} \div \rootsigma + 2(j_1-1) \rootsigma = \alpha_1 + 2(j_1-1) \rootsigma$.
From this we can calculate $\alpha_1$ with a $\pmod{\rootsigma}$ operation.
Similarly, the color of $e_2$ is $(\A_{r_2j_1} \mod \rootsigma + 2(j_1-1) \rootsigma +
\rootsigma= \beta_2 + 2j_1-1 \rootsigma$, from which we can calculate $\beta_2$.
Now, since $\A_{r_1 j_1} = \D + \A_{r_2 j_1}$, we have the following equality:
$\alpha_1\rootsigma+\beta_1 \:=\: \alpha_2\rootsigma+\beta_2+\alpha\rootsigma+\beta$.
From this equation, we can calculate $\beta_1$ by performing a $\pmod{\rootsigma}$ operation
on both sides.
We therefore have $\beta_1 \equiv  \beta_2+\beta \pmod{\rootsigma}$.
Once we know both $\alpha_1$ and $\beta_1$, we have $\A_{r_1 j_1}$.

Now, suppose the rank of the root node $u$ is $r = \R(u)$. By applying the Ranking Lemma
again, we can compute the difference $r_1 - r$. By Lemma~\ref{lem:rank2A}, this difference
along with the color sequence of $W$ gives us $\A_{r_1 j_1}-\A_{r j_1}$, from which
we can calculate $\A_{rj_1}$.
By repeating this argument for each of the $t$ dimension spanned by $W$, we can
calculate $\A_{r j_1}, \A_{r j_2}, \ldots, \A_{r, j_t}$. By the property of orthogonal arrays,
the ordered $t$-tuple $(\A_{r j_1}, \A_{r j_2},\dots, \A_{r j_t})$ is unique in $\A$.
Therefore, the rank of the root node $u$ is uniquely determined, which in turn
uniquely localizes $u$ in the lattice graph $G(n,d)$. This completes the proof.
\end{proof}
}

\section{Observability in Undirected Lattices}
\label{sec:Undir}

Observing a walk is more complicated in undirected graphs because edge colors fail
to determine the direction of the walk. An undirected edge $e = (u,v)$ may be traversed
in either direction by the walk, but reveals the same color. We say that a walk has
dimensions $t$ if it contains edges parallel to $t$ distinct dimensions. A minor
modification of the lower bound construction (Theorem~\ref{thm:lb}) shows that $\Omega (n^{t/d})$
colors are necessary for $t$-observability of undirected lattices: there are $(n/2)^d$
undirected paths, each of length $t \leq d$, and $k$ colors can disambiguate at most $k^t$
paths, giving the lower bound. We now prove a matching upper bound for $t$-observability.
As we mentioned earlier, there is one trivial walk that cannot be observed
in undirected graphs: a walk that traverses a single edge back and forth.
In this case, the current node cannot be determined from the color sequence.
However, we show below that walks that include at least two distinct edges
can always be observed.

\subsection{Signs of Undirected Edges and an Auxiliary Coloring}

Although the edges of $G(n,d)$ are undirected, a walk imposes an orientation on the
edges it visits. To exploit this induced directionality, we introduce the notion of a
\emph{sign}. First, we consider a walk as a \emph{sequence of nodes} $W = (u_0, u_1,
\ldots, u_\l )$, where $e_i = (u_{i-1}, u_i )$ is the $i$th edge in the walk.
Thus, the observed sequence is the colors of edges $e_1, e_2, \ldots, e_\l$.
Next, recall that a node $u$ with coordinates $(x_1,x_2,\ldots,x_d)$ has \emph{node rank}
$\R(u) = \sum_{i=1}^d  x_i n^{d-i}$.
Now, consider an edge $e_i = (u_{i-1}, u_i)$, and assume that it is parallel to the
dimension $j$. Then, it is easy to see that $\R(u_i)-\R(u_{i-1})=\pm n^{d-j}$.
We say that the \emph{sign} of the edge $e_i$ is \emph{positive} if
$\R(u_i)-\R(u_{i-1}) = +n^{d-j}$, and \emph{negative} otherwise.
(Intuitively, the sign is positive if the walk traverses the edge in the positive direction
of the axis, and negative otherwise.)

We first show that the signs of the edges in a walk can be computed from a simple
3-coloring. Let $o=(0,0,\ldots,0)$ be the origin and $u=(x_1,x_2,\ldots,x_d)$ be a
node of $G(n,d)$. Let $d_o(u) = \sum_i^d x_i$ be the length of the shortest path from
the origin to $u$. Then, {\sc mod3$(o)$} colors the edges of $G(n,d)$ as follows.
See Figure~\ref{fig:fig4} for an example.

\begin{quote}
	\emph{Assign color $m \;=\; d_o(u) \bmod 3$ to all the edges rooted at node $u$. }
\end{quote}

\begin{figure}[htbp]
\centering
\includegraphics[angle=0, width=0.25\textwidth]{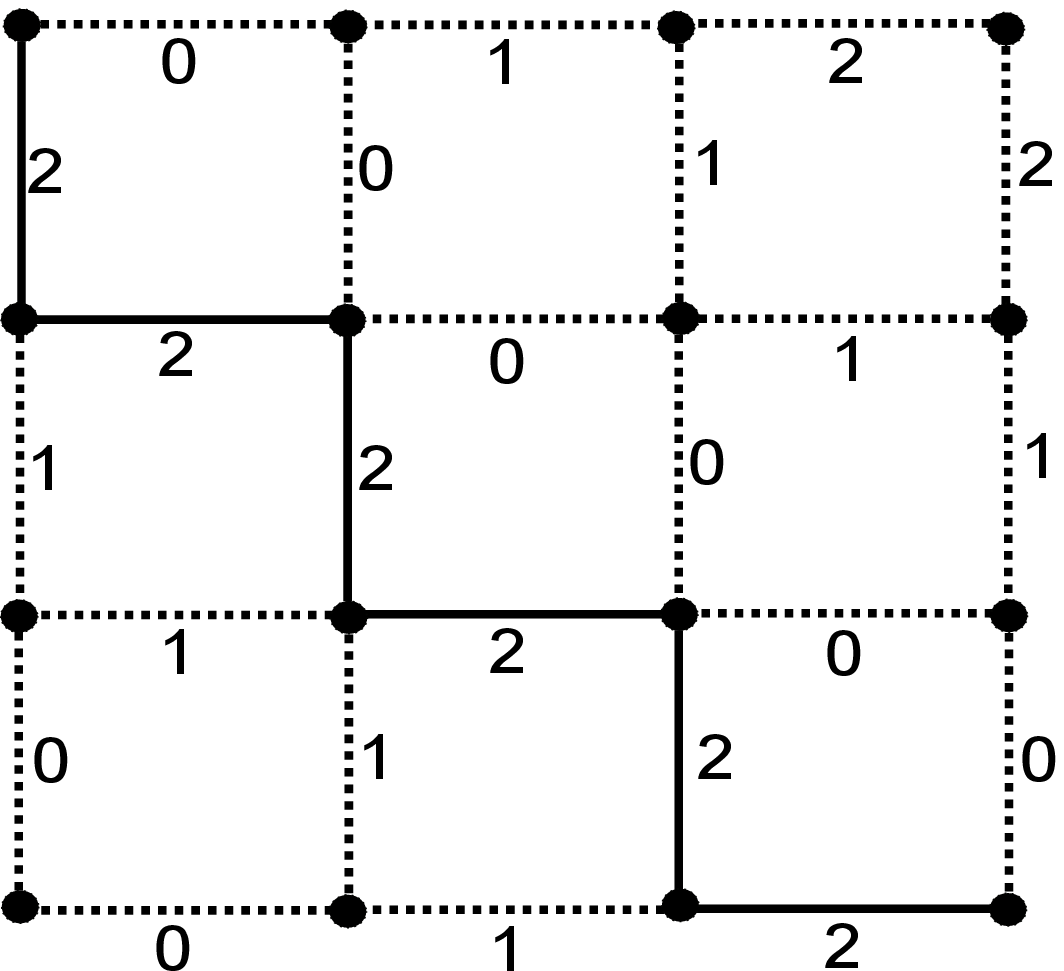}
\caption{{\sc mod3} coloring. A monochromatic walk using color 2 is shown as solid.}
\label{fig:fig4}
\end{figure}

\begin{lemma} \label{lem:trace1}
Let $G(n,d)$ be 3-colored using {\sc mod3$(o)$}, and let $W = (u_0, u_1, \ldots, u_\l )$ be
a walk with at least two distinct observed colors. Then we can compute the
sign of all $(u_{i-1}, u_i)$, for $i=1, 2, \ldots, \l$.
\end{lemma}
\begin{proof}
Assume, without loss of generality, that the walk includes two consecutive edges
$e_i = (u_{i-1}, u_i)$ and $e_{i+1} = (u_i, u_{i+1})$ with distinct colors
$c_i$ and $c_{i+1}$ under {\sc mod3$(o)$}.
We observe the following relationship between the colors and signs.
When $e_i$'s sign is positive, namely, $d_o(u_{i-1})=d_o(u_i)-1$, we have
$c_{i+1}=c_i+1 \bmod 3$ if $e_{i+1}$'s sign is positive, and
$c_{i+1}=c_i \bmod 3$ if $e_{i+1}$'s sign is negative.
Similarly, when $e_i$'s sign is negative, namely, $d_o(u_i)=d_o(u_{i-1})+1$, we have
$c_{i+1}=c_i \bmod 3$ if $e_{i+1}$'s sign is positive, and $c_{i+1}=c_i-1 \bmod 3$ otherwise.
Because $c_i \neq c_{i+1} \mod 3$, we conclude that $e_{i+1}$'s sign is positive
if $c_{i+1}=c_i + 1$, and negative otherwise.
Once the sign of one edge in the walk is determined, we can repeat this process
to infer the signs of all other edges.
\end{proof}
When the walk includes only edges of one color, we try
{\sc mod3} with $d-t+1$ other choices of the \emph{origin}. In particular, let {\sc mod3$(o_j)$}
be the 3-coloring with respect to the origin $o_0=o$ and origin $o_j = (0, 0, \ldots, n-1, \ldots, 0)$ with
$j$th coordinate $n-1$ and zeroes everywhere else, for $j=1,2, \ldots, d-t+1$.

\begin{lemma} \label{lem:trace2}
Given any $t$-dimensional walk $W=(u_0, u_1, \dots, u_\l)$ in $G(n,d)$ that visits at least two edges, there is
a 3-coloring {\sc mod3$(o_j)$} for which $W$ is \emph{not} monochromatic, for $j=0,1, \ldots, d-t+1$.
\end{lemma}
\begin{proof}
By the pigeon principle, a $t$-dimensional walk must contain at least one edge that parallels to one of the first $d-t+1$ axis. Assume, without loss of generality, that the walk includes two consecutive edges $e_i=(u_{i-1}, u_i)$ and $e_{i+1}=(u_i, u_{i+1})$ and edge $e_i$ parallels to the $j$th dimension, $1\leq j\leq d-t+1$. Let $c_i^j$ denote the color for edge $e_i$ using 3-coloring {\sc mod3$(o_j)$}. We show in Figure~\ref{fig:fig5} the color of edges $e_i$ and $e_{i+1}$ are distinct under 3-coloring {\sc mod3$(o_0)$} or {\sc mod3$(o_j)$}, namely, either $c_i^0\neq c_{i+1}^0$ or $c_i^j\neq c_{i+1}^j$ holds. (1) When $e_{i}$ roots at $u_{i-1}$ and $e_{i+1}$ roots at $u_i$, we have $c_{i+1}^0=c_{i}^0+1 \bmod 3$. Similarly (2) when $e_{i}$ roots at $u_{i}$ and $e_{i+1}$ roots at $u_{i+1}$, we have $c_{i+1}^0=c_{i}^0-1 \bmod 3$. On the other hand, (3) when both $e_i$ and $e_{i+}$ roots at $u_{i}$, we have $c_{i+1}^j=c_{i}^j+1 \bmod 3$. And at last, (4) when $e_{i}$ roots at $u_{i-1}$ and $e_{i+1}$ roots at $u_{i+1}$, we have $c_{i+1}^j=c_{i}^j-1 \bmod 3$.
\end{proof}

\begin{figure}[htbp]
\centering
\includegraphics[angle=0, width=0.6\textwidth]{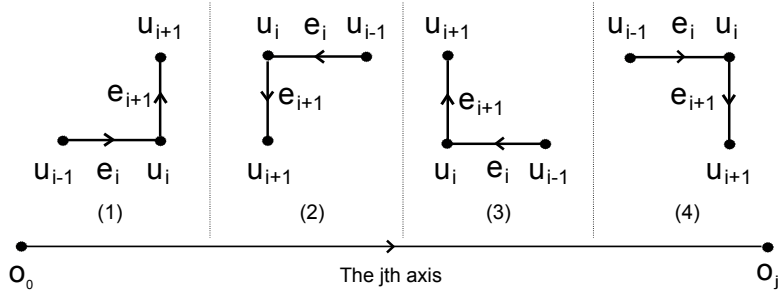}
\caption{Demonstration of 4 situations in Lemma~\ref{lem:trace2}.}
\label{fig:fig5}
\end{figure}

\subsection{Coloring Scheme for Undirected $t$-Observability}

Our final coloring scheme combines the orthogonal array based coloring with the
{\sc mod3} coloring. We use the orthogonal array $\A$ with $\s^t \geq n^d$, whose
$i$th row is used to color edges rooted at node with rank $i$. The algorithm uses
$2^t3^{d-t+2}$ groups of disjoint colors, $C_j$, each with $d \sigma$ colors, where
$0 \leq j < 2^t3^{d-t+2}$, and $\s = O(n^{d/t})$.

\paragraph{{\sc UndirColor}$(u)$:}

\begin{itemize}	\advance\itemsep by -4pt

\item Let $(a_1, a_2, \ldots, a_t )$ be the unique coefficient vector of a polynomial whose
	$p$-index $i$ equals the rank of $u$, namely, $i = \R(u)$. Let $(x_1, x_2, \ldots, x_d)$ be the coordinates of $u$.
\\
\item Let $m_1 \;=\; \sum_{k=1}^t (a_k \bmod 2) 2^{t-k}$.
\\
\item Let $m_2 \;=\; \sum_{j=1}^{d-t+1} ((\sum_{k=1}^d x_k) + n - 2x_j \bmod 3) 3^{j} + (\sum_{k=1}^d x_k  \bmod 3)$.
\\
\item Let $m \;=\; 2^t m_2 + m_1$.
\\
\item An edge $e$ rooted at $u$ gets color $A_{ij} \:+\: (j-1) \sigma$
	from the color group $C_m$ if $e$ is parallel to dimension $j$.

\end{itemize}

We can now prove the following result for
undirected $t$-dimensional walks in $G(n,d)$.

\begin{theorem} \label{thm:Undir}
Given an undirected lattice graph $G(n,d)$, we can color its edges with $O(n^{d/t})$ colors so that any $t$-dimensional walk that visits at least two distinct edges of $G(n,d)$ can be observed, where $d$ is a constant and $t \leq d$.
\end{theorem}
\begin{proof}
We color the graph $G(n,d)$ using {\sc UndirColor}, and show how to localize a walk $W$ of dimension $t$. Suppose that the color of edge $e_i$ belongs to color group $C_{m_i}$. Let $m_{i1}=m_i \bmod 2^t$, $m_{i2}=m_i\div 2^t$ and consider a ternary representation of $m_{i2}$, namely, $m_{i2}=\sum_{j=0}^{d-t+1}c_i^j3^j$. Then $c_i^j$ is the color of edge $e_i$ using 3-coloring {\sc mod3$(o_j)$}, $j=0,1,\dots,d-t+1$. By Lemma~\ref{lem:trace2} there is a 3-coloring {\sc mod3$(o_j)$} for which $W$ is not monochromatic. Then we can compute the sign of all edges in $W$ by using Lemma~\ref{lem:trace1}. Now for any two nodes in this walk, we could using the sign of all edges on the path between them to compute their rank different. Therefore, similar argument in Theorem~\ref{thm:upperD} can be applied to uniquely localize the root node of $W$ in lattice graph $G(n,d)$.
\end{proof}
\cut{

Similarly, the $i$th row of the array $\A$ is used to assign colors to the $d$ edges
rooted at the node with rank $i$, while the column $j$ is used to assign colors for all
edges parallel to dimension $j$. The use of disjoint color groups $C_m$ is required to detect the position of the root node ($m_1$) in the lattice and to decode the position of all other nodes ($m_2$) from the root node of any walk.
%
The following lemma easily follows from the assignment.

\begin{lemma}
{\sc ColorDU} assigns colors to all the edges of $G(n,d)$, uses $O(n^{d/t})$
colors, and no two outgoing edges of a node receive the same color.
\end{lemma}

We now discuss the main differences between directed and undirected graph, which is the proof of the Tracing Lemma:

\begin{lemma} {[$d$-Dimensional Tracing Lemma]} \label{lem:tracingdu}
Let $G(n,d)$ be colored using the scheme ${\sc ColorDU}$, and let $W$ be a walk in $G(n,d)$.
Fixing the position of any node of $W$ leads to a \emph{unique} embedding of $W$ in $G(n,d)$.
\end{lemma}

It is simple to prove the tracing lemma previously because the color sequence of a walk on directed graph gives us the orientations of all edges in this walk. However, it is not the case for undirected graph as edges do not have orientations. To deal with this problem, we use a node sequence instead of edge sequence to represent a walk in undirected graph. To be specific, a length $l$ walk $\W=\{u_1, u_2, \ldots, u_{l+1}\}$, if the $i$th node that $\W$ passes is $u_i$. Let $e_i$ be the edge between nodes $u_i$ and $u_{i+1}$. As the color of $e_i$ gives us its dimension, say the $j$th dimension, we have $\R(u_i)-\R(u_{i+1})=\pm n^{d-j}$. We define the sign of consecutive nodes pair $u_i$ and $u_{i+1}$ is positive if $\R(u_i)-\R(u_{i+1})=n^{d-j}$; negative if $\R(u_i)-\R(u_{i+1}) = - n^{d-j}$. Note that each consecutive nodes pair corresponding to a unique edge $e_i$, so the following proofs we may use the sign of edge $e_i$ represents the sign of nodes pair $u_i$ and $u_{i+1}$.  Clearly, this sign is positive (resp. negative) if the walk from node $u_i$ to $u_{i+1}$ goes toward the positive (resp. negative) direction. Therefore, if we could determine this sign of each consecutive nodes pair in this walk, fixing the position of any node of $\W$ leads to a unique embedding.

We build the following lemmas and coloring scheme dealing with this problem. As the whole coloring scheme is complicated, we start from a simpler case of a 3-color scheme. That is we use only 3 colors to color all edges in this graph in order to detect a unique embedding when fixing the position of any node in some walks. This will help build intuition for the general coloring.

Let $o=(0,0,\ldots,0)$ be the origin and $u=(x_1,x_2,\ldots,x_d)$ be any nodes in graph $G(n,d)$. Let $d_o(u)$ denotes the length of shortest path between nodes $o$ and $u$, then clearly $d_o(u)=\sum_i^d x_i$. Consider the following algorithm named ${\sc ColorD3}(o)$ using 3 colors $\{0, 1, 2\}$: For any node $u$ in graph $G(n,d)$, give color $m \;=\; d_o(u) \bmod 3$ to all edges rooted at $u$. Figure~\ref{fig:fig4} below show a 2-dimensional graph colored with ${\sc ColorD3}(o)$. Clearly all the edges are assigned one of the 3 colors $\{0, 1, 2\}$. The following lemma is crucial to detect embedding from the color sequence.

\begin{figure}[htbp]
\centering
\includegraphics[angle=0, width=0.5\textwidth]{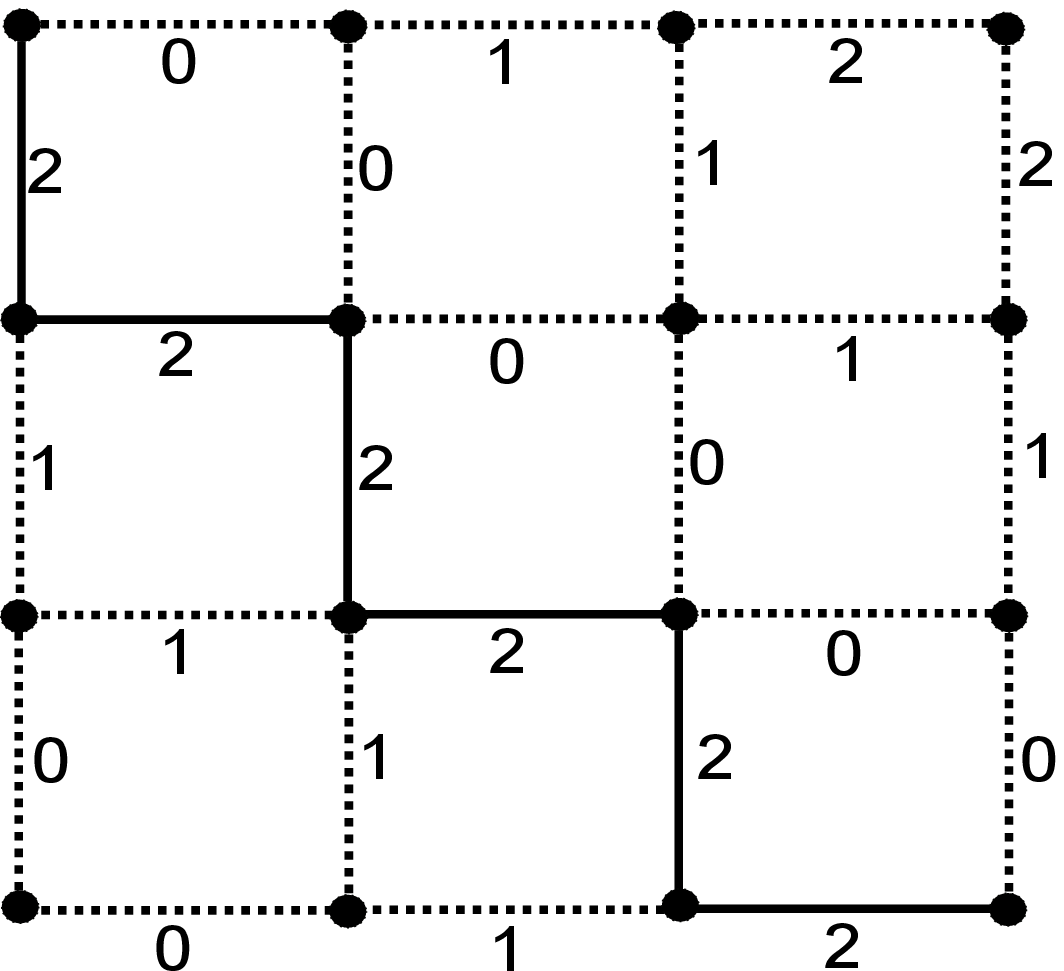}
\caption{A 2-dimensional lattice graph colored with ${\sc ColorD3}(o)$.
 The color of each edge is indicated by the number on its root node.}
\label{fig:fig4}
\end{figure}

\begin{lemma} \label{lemma:trace1}
Let $G(n,d)$ be colored using the scheme ${\sc ColorD3}(o)$, and let $W$ be a walk in $G(n,d)$. If there exists two consecutive edges $e_i$ and $e_{i+1}$ in this walk whose colors are distinct, then the sign of all consecutive node pairs in this walk can be inferred from the color sequence.
\end{lemma}
\begin{proof}
We prove the lemma by two steps. First we show that if there exist two consecutive edges, $e_i$ and $e_{i+1}$, with different color, then the sign of node pair $u_i$ and $u_{i+1}$ can be determined, where $u_i$ and $u_{i+1}$ are the root nodes of $e_i$ and $e_{i+1}$ respectively. Second we prove that the sign of node pair $u_i$ and $u_{i+1}$ leads to a unique sign on all other node pairs.

Let the color of $e_i$ and $e_{i+1}$ be $c_i$ and $c_{i+1}$, respectively, where $c_i\neq c_{i+1} \bmod 3$. We illustrate all situations of the sign of $e_i$ and $e_{i+1}$ being positive of negative. 1) When the sign of $e_i$ is positive, we have $d_o(u_i)=d_o(u_{i+1})-1$. Therefore if the sign of $e_{i+1}$ is positive, then $c_{i+1}=c_i+1 \bmod 3$; otherwise if the sign of $e_{i+1}$ is negative, then $c_{i+1}=c_i \bmod 3$. 2) When the sign of $e_i$ is negative we have $d_o(u_i)=d_o(u_{i+1})+1$. Therefore if the sign of $e_{i+1}$ is positive, then $c_{i+1}=c_i \bmod 3$; otherwise if the sign of $e_{i+1}$ is negative, then $c_{i+1}=c_i-1 \bmod 3$. As $c_i\neq c_{i+1} \bmod 3$, the sign of $e_i$ is an positive (resp. negative) if and only if $c_{i+1}=c_i + 1 \bmod 3$ (resp. $c_{i+1} = c_i - 1 \bmod 3$).

Similarly, if the sign of one edge is determined. By illustrating all possible situations of its consecutive next edge, we can detect its sign. We take the sign of $e_1$ being positive as an example. Then the color of the next edge $c_2=c_1 \bmod 3$ indicates that the sign of $e_2$ is negative; the color of the next edge $c_2=c_1+1 \bmod 3$ represents that the sign of $e_2$ is positive. This process can be applied for the previous edge of $e_1$, too. We can, therefore, infer the sign of all edges in this walk from its color sequence.
\end{proof}

Now the only situation we cannot handle is when the color sequence of a walk contains only one color. Note that our color scheme ${\sc ColorD3}(o)$ is established using the origin node $o$. This allows us to apply similar $3$-color scheme based on other node rather than the origin node. If when the color sequence of any walk under color scheme ${\sc ColorD3}(o)$ contains only when color, we could guarantee that there exists another node $u$ such that the color sequence of a walk under color scheme ${\sc ColorD3}(u)$ contains more than one color, then we are able to detect the orientations of the edges for any walk. The following lemma shows one possible selection of such nodes $u$.

\begin{lemma} \label{lemma:trace2}
Let $G(n,d)$ be colored using the scheme ${\sc ColorD3}(o)$, and let $W$ be a $t$ dimensional walk in $G(n,d)$. If the color sequence of $W$ under color scheme ${\sc ColorD3}(o)$ contains only one color, then there exists one node $a_i=(0,\ldots,0,n,0,\ldots,0)$ with the $i$th coordinate being $n$ and all other coordinates being $0$, $i=1,2,\ldots,d-t+1$, such that color sequence of $W$ under color scheme ${\sc ColorD3}(a_i)$ contains more than one color.
\end{lemma}
\begin{proof}
By the pigeon principle, a $t$ dimensional walk must contain at least one edge parallel to one of the first $d-t+1$ axis. We denote this edge by $e$ and let its dimension be $i$. Without loss of generality we assume that the sign of $e$ is positive. As the color sequence of $W$ under color scheme ${\sc ColorD3}(o)$ contains only one color, then every other edges in walk $W$ should move away from node $o$ and all other edges must move towards $o$. Note that in color scheme ${\sc ColorD3}(a_i)$ we compute the distance to node $a_i$ instead of using $o$. Then in color scheme ${\sc ColorD3}(a_i)$, $e$ is now moving towards node $a_i$. Then the previous edge of $e$, the next edge of $e$ and $e$ must all move towards $a_i$, which should be colored in different 3 colors. This complete the proof. [Prof. Suri, I think we may need to use other way to prove this thm. I tried several ways, and this is the best I got so far.]
\end{proof}

Now we can go back to prove the tracing lemma for undirected graph $G(n,d)$.

\begin{proof}
Let the edges in $W$ be $(e_1,e_2,\ldots,e_l)$ and let the colors of $e_i$ come from color group $m'_i$, $i=1,2,\ldots,l$. Let $a_0=o=(0,0,\ldots,0)$ be the origin node and $a_i=(0,\ldots,0,n,0,\ldots,0)$ with the $i$th coordinate being $n$, where $i=1,2,\ldots,d-t+1$. First we decode the sign information from $m'_i$ by modula it by $2^t$, namely $m_i=m'_i \bmod 2^t, \forall i$. Then consider the ternary representations of $m_i$, namely $m_i=\sum_{k=1}^{d-t+1} b_k^i 3^k$. Clearly $b_k^i\in \{0,1,2\}$ and $b_k^i$ represents the sign of edge $e_i$ from color scheme ${\sc ColorD3}(a_k)$, where $0\leq k\leq d-t+1$ and $1\leq i\leq d$. Then by lemma (\ref{lemma:trace1}) and (\ref{lemma:trace2}), either we could find the detect the signs of all edges from $b_0^i$, $1\leq i\leq d$, or there exist one dimension $k$ such that we could find the detect the signs of all edges from $b_k^i$, $1\leq i\leq d$. This complete the proof.
\end{proof}

Knowing the orientation of all edges, the following technical lemmas about the coloring scheme can be proved similarly as we did in the previous chapter.

\begin{lemma}	\label{lem:differenceu}
Let $p_1, p_2$, respectively, denote the $p$-indices of polynomials with coefficient
vectors $(a_1, a_2, \ldots, a_t)$ and $(b_1, b_2, \ldots, b_t)$.
Then, the quantity $A_{p_1 j} - A_{p_2 j}$ can be uniquely calculated, for any $j$
from the $t$ coordinate \emph{differences}, namely, $a_k-b_k$, for $k=1,2,\dots,t$.
\end{lemma}

\begin{lemma}	\label{lem:parityu}
Let $(a_1, a_2, \ldots, a_t )$ and $(b_1, b_2, \ldots, b_t )$ be coefficient vectors
with $p$-indices $p_1, p_2$, for $p_1 > p_2$, respectively. Let $\ell = p_1-p_2$, and suppose
$(c_1, c_2, \ldots, c_t )$ is the coefficient vector of the polynomial with $p$-index $\ell$.
Then, for each $k = 1,2, \ldots, t$, we either have $a_k-b_k \equiv c_k \pmod{\s}$, or we
have $a_k-b_k \equiv c_k+1 \pmod{\s}$.
\end{lemma}

\begin{lemma}	\label{lem:rank2Au}
Let $G(n,d)$ be colored using the scheme {\sc ColorD}, and let $W$ be a walk in $G(n,d)$.
Let $e_1$ and $e_2$ be two edges in this walk rooted, respectively, at $u_1$ and $u_2$
with ranks (lexicographic order) $\R(u_1) = r_1$ and $\R(u_2) = r_2$, with $r_1 < r_2$.
Then, the difference $r_1-r_2$ along with the colors of $e_1$ and $e_2$ are sufficient
to compute $A_{r_1 j} - A_{r_2 j} \pmod{\s}$.
\end{lemma}

\begin{lemma} {[Ranking Lemma}] 	\label{lem:ranku}
Let a lattice graph $G(n,d)$ be colored using the algorithm {\sc ColorD}, and let
$u_1, u_2$ be two nodes in a walk $W$ of $G(n,d)$. We can calculate the difference of
their ranks $\R(u_1) - \R(u_2)$ from just the color sequence of $W$.
\end{lemma}

We are now ready to prove our main theorem for undirected graph.

%

\begin{theorem} \label{thm:upperDu}
$O(n^{d/t})$ colors suffice for $t$-observability of any lattice graph $G(n,d)$,
for any fixed dimension $d$ and $t \leq d$.
\end{theorem}
\begin{proof}
We color the graph $G(n,d)$ using {\sc ColorDu}, and show how to localize a
walk $W$ of cycle dimension $t$. Without loss of generality, let
$j_1,j_2,\dots,j_t$, where $1\leq j_1 < j_2<\dots<j_t\leq d$, be the dimensions
spanned by the cycles of $W$. Let $e$ be a $j_1$-dimensional \emph{edge} in $W$
with rank $\R(e)=r_1$. Then the color of $e$ give us the quantity of $\A_{r_1 j_1}$.

Now, suppose the rank of the root node $u$ is $r = \R(u)$. By applying the Ranking Lemma
again, we can compute the difference $r_1 - r$. By Lemma~\ref{lem:rank2Au}, this difference
along with the color sequence of $W$ gives us $\A_{r_1 j_1}-\A_{r j_1}$, from which
we can calculate $\A_{rj_1}$.
By repeating this argument for each of the $t$ dimension spanned by $W$, we can
calculate $\A_{r j_1}, \A_{r j_2}, \ldots, \A_{r, j_t}$. By the property of orthogonal arrays,
the ordered $t$-tuple $(\A_{r j_1}, \A_{r j_2},\dots, \A_{r j_t})$ is unique in $\A$.
Therefore, the rank of the root node $u$ is uniquely determined, which in turn
uniquely localizes $u$ in the lattice graph $G(n,d)$. This completes the proof.
\end{proof}

\section{A Stronger Bound for Directed Lattice Graphs}
\label{sec:orient}

Our analysis in Sections~\ref{sec:defs}--\ref{sec:upper-dir} requires that, for a
dimension to be spanned, the walk must include edges of \emph{both} orientation for that
dimension. In this section, we show a stronger result where such \emph{paired orientations}
are not necessary. In particular, we say that a walk in a directed lattice graph $G(n,d)$
is \emph{$t$-oriented} if it includes edges with $t$ distinct orientations; clearly, the
number of distinct orientation is at most $2d$. It is easy to see that $\Omega (n^{d/t})$
colors are necessary for the observability of $t$-oriented walks: just consider the
directed paths using the first $t$ edges in the construction of Theorem~1.
The following theorem shows that this bound is tight.

\begin{theorem}
Given a directed lattice graph $G(n,d)$, we can color its edges with $O(n^{d/t})$ colors so
that any $t$-oriented walk can be observed, where $d$ is a constant and $t \leq d$.
\end{theorem}

}

\cut{

In this section we introduce the new color scheme for directed lattice graph. We use $G(n,d)$ denote a directed lattice graph of size $n^d$.
Clearly such graph $G(n,d)$ has $2t$ orientations (two for each dimensions). Here we define a directed walk $W$ in $G(n,d)$ spans (resp. negative) orientation $2i-1$ (resp. $2i$) if it includes positive (resp. negative) edges along dimension $i$. Then naturally, a walk is said to be \emph{$t$-oriented} if it spans $t$ distinct orientations, for $t \leq 2d$. Edge in a given directed walk has a natural orientation: it is either in the positive or the negative direction. To be able to refer to this directionality, we call an edge \emph{$j$-up-edge} (resp., \emph{$j$-down-edge}) if it has \emph{positive} (resp. negative) orientation in $j$th dimension in this walk. Then directed graph $G(n,d)$ is called \emph{$t$-observable} if an agent can uniquely determine its current position in the graph from the observed edge color sequence of
\emph{any} $t$-oriented walk, for $t \leq 2d$.

We index all nodes of $G(n,d)$ as we did in previous section. And we use the same orthogonal array $\A$ satisfies $\s^t \geq n^d$, we uniquely associate the node with rank $i$ to the $i$th row of $\A$. We can now describe our coloring scheme in directed graph. Let $u$ be a node of $G(n,d)$ whose rank (lexicographic order) is $\R(u) = i$, where $0 \leq i < n^d$. Then, we use the $i$th row of $\A$ to assign to colors to the
edges rooted at $u$. The rules for assigning colors are described in the following algorithm. The algorithm uses $2^t$ groups of disjoint colors $C_0, C_1, \ldots, C_{2^t -1}$, each with $d \sigma$ colors.

\paragraph{{\sc ColorD2}$(u)$:}

\begin{itemize}	\advance\itemsep by -4pt

\item Let $(a_1, a_2, \ldots, a_t )$ be the unique coefficient vector of a polynomial whose
	$p$-index $i$ equals the rank of $u$, namely, $i = \R(u)$.

\item Let $m \;=\; \sum_{k=1}^t (a_k \bmod 2) 2^{t-k-1}$.

\item If $e$ spans the $j$th orientation, then give it color \\
        $A_{ij} \:+\: (j-1) \sigma$ from the color group $C_m$.

\end{itemize}

The following lemmas easily follow from the assignment.

\begin{lemma}
${\sc ColorD2}$ assigns colors to all the edges of $G(n,d)$, uses $O(n^{d/t})$
colors, and no two outgoing edges of a node receive the same color.
\end{lemma}

\begin{lemma} {[$d$-Dimensional Tracing Lemma]} \label{lem:tracingdu}
Let $G(n,d)$ be colored using the scheme ${\sc ColorD2}$, and let $W$ be a walk in $G(n,d)$.
Fixing the position of any node of $W$ leads to a \emph{unique} embedding of $W$ in $G(n,d)$.
\end{lemma}

Knowing the orientation of all edges, the following technical lemmas about the coloring scheme can be proved similarly as we did in the previous chapter.

\begin{lemma}	\label{lem:differenceu}
Let $p_1, p_2$, respectively, denote the $p$-indices of polynomials with coefficient
vectors $(a_1, a_2, \ldots, a_t)$ and $(b_1, b_2, \ldots, b_t)$.
Then, the quantity $A_{p_1 j} - A_{p_2 j}$ can be uniquely calculated, for any $j$
from the $t$ coordinate \emph{differences}, namely, $a_k-b_k$, for $k=1,2,\dots,t$.
\end{lemma}

\begin{lemma}	\label{lem:parityu}
Let $(a_1, a_2, \ldots, a_t )$ and $(b_1, b_2, \ldots, b_t )$ be coefficient vectors
with $p$-indices $p_1, p_2$, for $p_1 > p_2$, respectively. Let $\ell = p_1-p_2$, and suppose
$(c_1, c_2, \ldots, c_t )$ is the coefficient vector of the polynomial with $p$-index $\ell$.
Then, for each $k = 1,2, \ldots, t$, we either have $a_k-b_k \equiv c_k \pmod{\s}$, or we
have $a_k-b_k \equiv c_k+1 \pmod{\s}$.
\end{lemma}

\begin{lemma}	\label{lem:rank2Au}
Let $G(n,d)$ be colored using the scheme {\sc ColorD2}, and let $W$ be a walk in $G(n,d)$.
Let $e_1$ and $e_2$ be two edges in this walk rooted, respectively, at $u_1$ and $u_2$
with ranks (lexicographic order) $\R(u_1) = r_1$ and $\R(u_2) = r_2$, with $r_1 < r_2$.
Then, the difference $r_1-r_2$ along with the colors of $e_1$ and $e_2$ are sufficient
to compute $A_{r_1 j} - A_{r_2 j} \pmod{\s}$.
\end{lemma}

\begin{lemma} {[Ranking Lemma}] 	\label{lem:ranku}
Let a lattice graph $G(n,d)$ be colored using the algorithm {\sc ColorD2}, and let
$u_1, u_2$ be two nodes in a walk $W$ of $G(n,d)$. We can calculate the difference of
their ranks $\R(u_1) - \R(u_2)$ from just the color sequence of $W$.
\end{lemma}

We are now ready to prove our main theorem for undirected graph.

\begin{theorem} \label{thm:upperDu}
$O(n^{d/t})$ colors suffice for $t$-observability of any lattice graph $G(n,d)$,
for any fixed dimension $d$ and $t \leq 2d$.
\end{theorem}
\begin{proof}
We color the graph $G(n,d)$ using {\sc ColorD2}, and show how to localize a
walk $W$ of cycle dimension $t$. Without loss of generality, let
$j_1,j_2,\dots,j_t$, where $1\leq j_1 < j_2<\dots<j_t\leq 2d$, be the orientations
spanned by $W$. Let $e$ be a $j_1$-oriented \emph{edge} in $W$
with rank $\R(e)=r_1$. Then the color of $e$ give us the quantity of $\A_{r_1 j_1}$.

Now, suppose the rank of the root node $u$ is $r = \R(u)$. By applying the Ranking Lemma
again, we can compute the difference $r_1 - r$. By Lemma~\ref{lem:rank2Au}, this difference
along with the color sequence of $W$ gives us $\A_{r_1 j_1}-\A_{r j_1}$, from which
we can calculate $\A_{r j_1}$.
By repeating this argument for each of the $t$ dimension spanned by $W$, we can
calculate $\A_{r j_1}, \A_{r j_2}, \ldots, \A_{r, j_t}$. By the property of orthogonal arrays,
the ordered $t$-tuple $(\A_{r j_1}, \A_{r j_2},\dots, \A_{r j_t})$ is unique in $\A$.
Therefore, the rank of the root node $u$ is uniquely determined, which in turn
uniquely localizes $u$ in the lattice graph $G(n,d)$. This completes the proof.
\end{proof}

}

\section{Concluding Remarks and Extensions}
\label{sec:concl}

In this paper, we explored an observability problem for lattice graphs, and presented
asymptotically tight bounds for $t$-observability of both directed and undirected graphs.
The bounds reveal an interesting dependence on the ratio between the dimension of the
lattice and that of the walk, \emph{the larger the dimension of the walk the
smaller the color complexity of observing it}, as well as an unexpected conclusion
that the color complexity for full-dimensional walks is independent of the
lattice dimension.

Our results are easy to generalize to non-square lattice graphs, albeit at the expense
of more involved calculations. In particular, given a lattice graph of size
$N = n_1 \times n_2 \times \cdots \times n_d$, the number of colors for $t$-observability is
$\Theta (\sqrt[t]{N})$. Briefly, we use an $(\s, t, d)$ orthogonal array with $\s$
as the smallest prime larger than $N^{1/t}$. The nodes of the lattice
are ranked in the lexicographic order of their coordinates, and we can calculate the rank
of a node $u$ with coordinates $(x_1,x_2,\dots x_d)$, $x_j\in \{0,1,\dots,n_j -1\}$,
as
\[
r (u) \;=\; x_1 n_2 n_3\cdots n_d \;+\: x_2 n_3 \cdots n_d + \dots +
	x_{d-1}n_d + x_d \:=\: \sum^{d}_{j=1}(x_j\Pi_{k=j+1}^d n_k) .\]
Except for these minor modifications, the coloring scheme remains unchanged.
Given a walk $W$, and two nodes $u_1, u_2$ in the walk, the rank distances is calculated
as follows: if the number of $j$-up and $j$-down edges in the walk from $u_1$ to $u_2$ is
$\alpha_j$ and $\beta_j$, then $r(u_1)-r(u_2)=\Sigma_{j=1}^d(\beta_j-\alpha_j)\Pi_{k=j+1}^d n_k$.
The remaining technical machinery does not depend on the square lattice, and carries over to
rectangular lattices.

A number of research directions and open problems are suggested by this research. Our
coloring scheme and proof techniques should extend to other regular but non-rectangular
lattices; we can show this for planar hexagonal lattices but have not explored the idea fully.
On the other hand, observability of general graphs, even planar graphs of bounded degree,
appears to be quite challenging. It will also be interesting to explore observability
under \emph{node-coloring}.

Finally, some of the \emph{small world graph} models are essentially lattice graphs with
few random long-range neighbors at each node~\cite{kleinberg,watts}. It will be
interesting to extend our results to those graphs.

\bibliographystyle{plain}
\bibliography{paper}

\newpage

\appendix

\section{Appendix: Omitted Proofs and Details}

\subsection{Omitted Proofs from Section 5}

\subsubsection{Proof of Lemma~\ref{lem:OA}}
\label{sec:a2}

The array has dimensions $\s^t \times d$, and its entries come from the set
$\{0, 1, \ldots, \s - 1\}$, by construction. Thus, we only need to show that within
any $t$ columns of $\A$, all rows are distinct. We prove this by contradiction.
Let $j_1,j_2,\dots,j_t$, for $1 \leq j_1 < j_2 < \dots < j_t \leq d$ be any $t$ columns
of $\A$, and suppose that two different rows with indices $i_1$ and $i_2$, for
$i_1 < i_2$ are identical over these columns.
Let $(a_1, a_2, \ldots, a_t)$ and $(b_1, b_2, \ldots, b_t)$ denote the coefficients
corresponding to polynomials used for rows $i_1$ and $i_2$, respectively. Then,
by Equation~(\ref{eq:OA}), the polynomial used to construct entries of row $i_1$ is
$\P_{i_1} : a_1 x^{t-1} + a_2 x^{t-2} + \cdots + a_t$, and the polynomial used to
construct entries of row $i_2$ is $\P_{i_2} : b_1 x^{t-1} + b_2 x^{t-2} + \cdots + b_t$.
If these rows are identical, then we must have
$\P_{i_1} (j_k) \equiv \P_{i_2} (j_k) \pmod{\s}$, for $k=1, 2, \ldots, t$.  This implies that
$j_1, j_2, \ldots, j_t$ are $t$ distinct roots of the equation
$\P_{i_1} (x) - \P_{i_2} (x) \pmod{\s}$, which is not possible since this polynomial
has degree $t-1$ and at most $t-1$ distinct roots.\footnote{%
        Finite fields belong to unique factorization domains, and therefore a polynomial of
        order $r$ over finite fields has a unique factorization, and at most $r$ roots.}
Therefore, the rows $i_1$ and $i_2$ are not identical over the chosen $t$ columns, proving
that $\A$ is an orthogonal array.

\subsubsection{Proof of Lemma~\ref{lem:parity}}
\label{sec:a3}

Because $\ell = p_1-p_2$, we have $\ell = \sum_{k=1}^t (a_k - b_k) \s^{t-k}$.
Perform a $\pmod{\s^{t-k+1}}$ operation on both sides of the equality:
$(a_1-b_1)\s^{t-1}+\dots+(a_t-b_t) \;\equiv\; c_1\s^{t-1}+\dots+c_t$.
We get $(a_k-b_k)\s^{t-k}+\dots+(a_t-b_t) \;\equiv \; c_k\s^{t-k}+\dots+c_t
\pmod{\s^{t-k+1}}$, which is
equivalent to $(a_k-b_k-c_k)\s^{t-k} \;\equiv\; \sum_{i=k+1}^t (c_i-a_i+b_i) \s^{t-i}
\pmod{\s^{t-k+1}}$.
Because $a_i, b_i, c_i \in GF(\s)$, the right hand side clearly satisfies
the following bounds:
\[ -\s^{t-k} \:<\: \sum_{i=k+1}^t (c_i-a_i+b_i) \s^{t-i} \:<\: 2\s^{t-k} . \]
But since $a_k,b_k,c_k$ on the left side are positive integers, there are only two feasible
solutions: $\sum_{i=k+1}^t (c_i-a_i+b_i)\s^{t-i} \equiv 0  \pmod{\s^{t-k+1}}$
or $\sum_{i=k+1}^t (c_i-a_i+b_i)\s^{t-i} \equiv \s^{t-k} \pmod{\s^{t-k+1}}$.
Therefore, we must have either
$a_k-b_k-c_k \equiv 0 \pmod{\s}$, or $a_k-b_k-c_k \equiv 1 \pmod{\s}$.
This completes the proof.

\subsubsection{Proof of Lemma~\ref{lem:rank2A}}
\label{sec:a4}

Suppose that the colors of $e_1$ and $e_2$ belongs to color groups $C_{m_1}$ and $C_{m_2}$,
respectively. Let us consider the binary representations of $m_1, m_2$, namely,
$m_1 = \sum_{k=1}^t (a_k \bmod 2) 2^{t-k-1}$ and $m_2 = \sum_{k=1}^t (b_k \bmod 2) 2^{t-k-1}$.
If $(a_k - b_k) \equiv 0 \pmod{2}$, then by Lemma~\ref{lem:parity} we can conclude
that $(a_k - b_k) \equiv c_k \pmod{\s}$ if $c_k$ is even; otherwise it equals $c_k +1$.
Similarly, $(a_k - b_k) \equiv 1 \pmod{2}$ tells us that
$(a_k - b_k) \equiv c_k \pmod{\s}$ if $c_k$ is odd (and $c_k +1$ otherwise).
Because we can infer from the embedding whether $a_k > b_k$, we can calculate
$a_k - b_k$ from $a_k - b_k \pmod{\s}$.
Once we know all $a_k - b_k$, for $k=1,2, \ldots, t$ we can
compute $A_{r_1 j} - A_{r_2 j} \pmod{\s}$
using Lemma~\ref{lem:difference}.

\subsection{Omitted Proofs from Section 6}

\subsubsection{Proof of Lemma~\ref{lem:trace1}}
\label{sec:trace1}

Assume, without loss of generality, that the walk includes two consecutive edges
$e_i = (u_{i-1}, u_i)$ and $e_{i+1} = (u_i, u_{i+1})$ with distinct colors
$c_i$ and $c_{i+1}$ under {\sc mod3$(o)$}.
We observe the following relationship between the colors and signs.
When $e_i$'s sign is positive, namely, $d_o(u_{i-1})=d_o(u_i)-1$, we have
$c_{i+1}=c_i+1 \bmod 3$ if $e_{i+1}$'s sign is positive, and
$c_{i+1}=c_i \bmod 3$ if $e_{i+1}$'s sign is negative.
Similarly, when $e_i$'s sign is negative, namely, $d_o(u_i)=d_o(u_{i-1})+1$, we have
$c_{i+1}=c_i \bmod 3$ if $e_{i+1}$'s sign is positive, and $c_{i+1}=c_i-1 \bmod 3$ otherwise.
Because $c_i \neq c_{i+1} \mod 3$, we conclude that $e_{i+1}$'s sign is positive
if $c_{i+1}=c_i + 1$, and negative otherwise.
Once the sign of one edge in the walk is determined, we can repeat this process
to infer the signs of all other edges.


\end{document}